\documentclass[11pt]{article}
\usepackage{graphics,latexsym}
\usepackage{citesort}

\usepackage{times}
\usepackage{latexsym}
\usepackage{epsfig}
\usepackage{url,wrapfig,setspace}
\usepackage{subfigure}
\usepackage{amsmath}

\usepackage[top=1.in, bottom=0.9in, left=1.0in, right=1.0in]{geometry}

\newtheorem{DummyLemma}{DummyLemma}[section]
\newtheorem{theorem}[DummyLemma]{\bf Theorem}
\newtheorem{lemma}[DummyLemma]{\bf Lemma}
\newtheorem{corollary}[DummyLemma]{\bf Corollary}
\newtheorem{ExamplE}[DummyLemma]{\bf Example}
\newtheorem{DefinitioN}[DummyLemma]{\bf Definition}

\newenvironment{proof}
               {\noindent {\bf Proof:} \rm}{\hfill $\Box$}

\begin{document}

\title{XML Reconstruction View Selection in XML Databases:\\Complexity Analysis and Approximation Scheme}
\author{Artem Chebotko and Bin Fu\\
Department of Computer Science\\
University of Texas-Pan American\\
Edinburg, TX 78539, USA\\
\{artem, binfu\}@cs.panam.edu\\
}

\maketitle

\begin{abstract}
Query evaluation in an XML database requires reconstructing XML
subtrees rooted at nodes found by an XML query. Since XML subtree
reconstruction can be expensive, one approach to improve query
response time is to use reconstruction views - materialized XML
subtrees of an XML document, whose nodes are frequently accessed by
XML queries. For this approach to be efficient, the principal
requirement is a framework for view selection. In this work, we are
the first to formalize and study the problem of XML reconstruction
view selection. The input is a tree $T$, in which every node $i$ has
a size $c_i$ and profit $p_i$, and the size limitation $C$. The
target is to find a subset of subtrees rooted at nodes $i_1,\cdots,
i_k$ respectively such that $c_{i_1}+\cdots +c_{i_k}\le C$, and
$p_{i_1}+\cdots +p_{i_k}$ is maximal. Furthermore, there is no
overlap between any two subtrees selected in the solution. We prove
that this problem is NP-hard and present a fully polynomial-time
approximation scheme (FPTAS) as a solution.
\end{abstract}

\begin{figure*}[!htb]
  \begin{center}
    \mbox{
      \subfigure[XML tree]{\scalebox{0.8}{\includegraphics{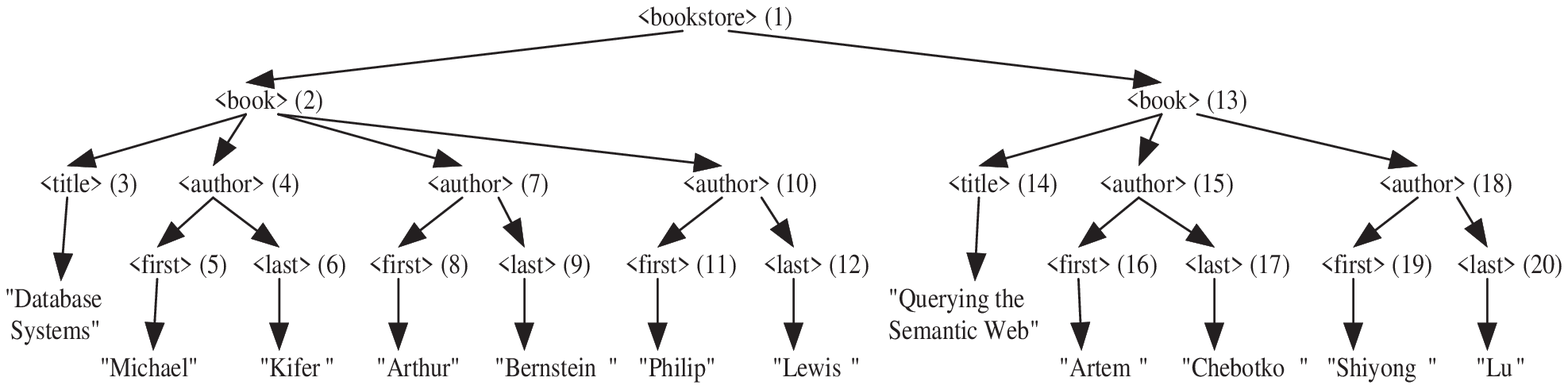} \label{fig:xmltree}}}
      }

    \mbox{
    \subfigure[Edge table]{
        \begin{minipage}[b]{0.5\linewidth}
        \footnotesize\centering

            \vspace{0.5cm}
            \begin{tabular}[t]{|l|l|l|l|}

            \multicolumn{4}{c}{\bf r$_{edge}$}\\
            \hline

            \underline{\textit{ID}} & \textit{parentID} & \textit{name} & \textit{content} \\
            \hline

            1 & \texttt{NULL} & bookstore & \texttt{NULL} \\
            \hline

            2 & 1 & book & \texttt{NULL} \\
            \hline

            3 & 2 & title & Database Systems \\
            \hline

            4 & 2 & author & \texttt{NULL} \\
            \hline

            5 & 4 & first & Michael \\
            \hline

            6 & 4 & last & Kifer \\
            \hline

            7 & 2 & author & \texttt{NULL} \\
            \hline

            8 & 7 & first & Arthur \\
            \hline

            9 & 7 & last & Bernstein \\
            \hline

            10 & 2 & author & \texttt{NULL} \\
            \hline

            11 & 10 & first & Philip \\
            \hline

            12 & 10 & last & Lewis \\
            \hline

            13 & 1 & book & \texttt{NULL} \\
            \hline

            14 & 13 & title & Querying the Semantic Web \\
            \hline

            ... & ... & ... & ... \\
            \hline

            \end{tabular}\vspace{0.1cm}

        \end{minipage}
        \label{fig:relation}
    }

\subfigure[XML reconstruction view]{
        \begin{minipage}[b]{0.5\linewidth}
        \footnotesize\centering

            \vspace{0.5cm}
            \begin{tabbing}
            \hspace*{0.2in}\=\hspace*{0.2in}\=\hspace*{0.2in}\=\hspace*{0.2in}\=\hspace*{0.2in}\=\hspace*{0.2in}\=\hspace*{0.2in}\=\\
            \texttt{<book>}\\
            \> \texttt{<title> Database Systems </title>}\\
            \> \texttt{<author>}\\
            \>\> \texttt{<first> Michael </first>}\\
            \>\> \texttt{<last> Kifer </last>}\\
            \> \texttt{</author>}\\
            \> \texttt{<author>}\\
            \>\> \texttt{<first> Arthur </first>}\\
            \>\> \texttt{<last> Bernstein </last>}\\
            \> \texttt{</author>}\\
            \> \texttt{<author>}\\
            \>\> \texttt{<first> Philip </first>}\\
            \>\> \texttt{<last> Lewis </last>}\\
            \> \texttt{</author>}\\
            \texttt{</book>}\\
            \\
            \end{tabbing}

        \end{minipage}
        \label{fig:view}
    }

    }

    \caption{An example of an XML tree, its relational storage, and XML reconstruction view}
    \label{fig:xml}
  \end{center}
\end{figure*}

\section{Introduction}

With XML\footnote{\url{http://www.w3.org/XML}}~\cite{BunemanSA99}
being the de facto standard for business and Web data
representation and exchange, storage and querying of large XML
data collections is recognized as an important and challenging
research problem. A number of XML
databases~\cite{FlorescuK99,ShanmugasundaramGTZDN99,BohannonFRS02,KrishnamurthyCKN04,ChaudhuriCSW05,BalminP05,BonczGKMRT06,AtayCLLF07,Schoning01,JagadishACLNPPSWWY02,HundlingSW05}
have been developed to serve as a solution to this problem. While
XML databases can employ various storage models, such as
relational model or native XML tree model, they support standard
XML query languages, called
XPath\footnote{\url{http://www.w3.org/TR/xpath}} and
XQuery\footnote{\url{http://www.w3.org/XML/Query}}. In general, an
XML query specifies which nodes in an XML tree need to be
retrieved. Once an XML tree is stored into an XML database, a
query over this tree usually requires two steps: (1)~finding the
specified nodes, if any, in the XML tree and (2)~reconstructing
and returning XML subtrees rooted at found nodes as a query
result. The second step is called XML subtree
reconstruction~\cite{ChebotkoALF07,ChebotkoLALF05} and may have a
significant impact on query response time. One approach to
minimize XML subtree reconstruction time is to cache XML subtrees
rooted at frequently accessed nodes as illustrated in the
following example.

Consider an XML tree in Figure~\ref{fig:xmltree} that describes a
sample bookstore inventory. The tree nodes correspond to XML
elements, e.g., \emph{bookstore} and \emph{book}, and data values,
e.g., \emph{``Arthur''} and \emph{``Bernstein''}, and the edges
represent parent-child relationships among nodes, e.g., all the
\emph{book} elements are children of \emph{bookstore}. In
addition, each element node is assigned a unique identifier that
is shown next to the node in the figure. As an example, in
Figure~\ref{fig:relation}, we show how this XML tree can be stored
into a single table in an RDBMS using the edge
approach~\cite{FlorescuK99}. The edge table $r_{edge}$ stores each
XML element as a separate tuple that includes the element ID, ID
of its parent, element name, and element data content. A sample
query over this XML tree that retrieves books with title
``Database Systems'' can be expressed in XPath as:

\vspace{-0.5cm}
            \begin{tabbing}
            \hspace*{0.2in}\=\hspace*{0.2in}\=\hspace*{0.2in}\=\hspace*{0.2in}\=\hspace*{0.2in}\=\hspace*{0.2in}\=\hspace*{0.2in}\=\\
            \> {\footnotesize \texttt{/bookstore/book[title="Database Systems"]}}\\
            \end{tabbing}\vspace{-0.5cm}
This query can be translated into relational algebra or SQL over
the edge table to retrieve IDs of the \emph{book} elements that
satisfy the condition:

\vspace{-0.5cm}
            \begin{tabbing}
            \hspace*{0.2in}\=\hspace*{0.2in}\=\hspace*{0.2in}\=\hspace*{0.2in}\=\hspace*{0.2in}\=\hspace*{0.2in}\=\hspace*{0.2in}\=\\
            \> $\pi_{r_2.ID}\ ($\\
            \>\> $r_1 \bowtie_{r_1.ID=r_2.parentID \wedge r_1.name=`bookstore' \wedge}$\\
            \>\>\>\> $_{r_1.parentID\ is\ NULL \wedge r_2.name=`book'}$\\
            \>\> $r_2 \bowtie_{r_2.ID = r_3.parentID \wedge r_3.name = `title' \wedge}$\\
            \>\>\>\> $_{r_3.content = `Database Systems'}$\\
            \>\> $r_3 \ )$\\
            \end{tabbing}\vspace{-0.5cm}
where $r_1$, $r_2$, and $r_3$ are aliases of table $r_{edge}$. For
the edge table in Figure~\ref{fig:relation}, the relational
algebra query returns ID ``2'', that uniquely identifies the first
\emph{book} element in the tree. However, to retrieve useful
information about the book, the query evaluator must further
retrieve all the descendants of the \emph{book} node and
reconstruct their parent-child relationships into an XML subtree
rooted at this node; this requires additional self-joins of the
edge table and a reconstruction algorithm, such as the one
proposed in~\cite{ChebotkoALF07}. Instead, to avoid expensive XML
subtree reconstruction, the subtree can be explicitly stored in
the database as an XML reconstruction view (see
Figure~\ref{fig:view}). This materialized view can be used for the
above XPath query or any other query that needs to reconstruct and
return the \emph{book} node (with ID ``2'') or its descendant.

In this work, we study the problem of selecting XML reconstruction
views to materialize: given a set of XML elements $D$ from an XML
database, their access frequencies $a_i$ (aka workload), a set of
ancestor-descendant relationships $AD$ among these elements, and a
storage capacity $\delta$, find a set of elements $M$ from $D$,
whose XML subtrees should be materialized as reconstruction views,
such that their combined size is no larger than $\delta$. To our
best knowledge, our solution to this problem is the first one
proposed in the literature. Our main contributions and the paper
organization are as follows. In Section~\ref{sec:RW}, we discuss
related work. In Section~\ref{sec:problem}, we formally define the
XML reconstruction view selection problem. In
Sections~\ref{sec:complexity} and \ref{sec:FPTAS}, we prove that
the problem is NP-hard and describe a fully polynomial-time
approximation scheme (FPTAS) for the problem. We conclude the
paper and list future work directions in
Section~\ref{sec:conclude}.

\section{Related Work} \label{sec:RW}

We studied the XML subtree reconstruction problem in the context
of a relational storage of XML documents
in~\cite{ChebotkoALF07,ChebotkoLALF05}, where several algorithms
have been proposed. Given an XML element returned by an XML query,
our algorithms retrieve all its descendants from a database and
reconstruct their relationships into an XML subtree that is
returned as the query result. To our best knowledge, there have
been no previous work on materializing reconstruction views or XML
reconstruction view selection.

Materialized
views~\cite{BalminOBCP04,MandhaniS05,XuO05,LakshmananWZ06,FanGJK07,TangYOCW08}
have been successfully used for query optimization in XML
databases. These research works rewrite an XML query, such that it
can be answered either using only available materialized views, if
possible, or accessing both the database and materialized views.
View maintenance in XML databases has been studied in
\cite{TatemuraSPAC05,SawiresTPAAC06}. There have been only one
recent work~\cite{TangYTOB09} on materialized view selection in
the context of XML databases. In~\cite{TangYTOB09}, the problem is
defined as: find views over XML data, given XML databases, storage
space, and a set of queries, such that the combined view size does
not exceed the storage space. The proposed solution produces
minimal XML views as candidates for the given query workload,
organizes them into a graph, and uses two view selection
strategies to choose views to materialize. This approach makes an
assumption that views are used to answer XML queries completely
(not partially) without accessing an underlying XML database. The
XML reconstruction view problem studied in our work focuses on a
different aspect of XML query processing: it finds views to
materialize based on how frequently an XML element needs to be
reconstructed. However, XML reconstruction views can be
complimentarily used for query answering, if desired.

Finally, the materialized view selection problem have been
extensively studied in data
warehouses~\cite{BaralisPT97,YangKL97,Gupta97,GuptaM99,KarloffM99,ChirkovaHS02}
and distributed databases~\cite{Kossmann00}. These research
results are hardly applicable to XML tree structures and in
particular to subtree reconstruction, which is not required for
data warehouses or relational databases.

\section{XML Reconstruction View Selection Problem}
\label{sec:problem}

In this section, we formally define the XML reconstruction view
selection problem addressed in our work.

\textbf{Problem formulation.} Given $n$ XML elements, $D = \{D_1,
D_2, \cdots, D_n\}$, and an ancestor-descendant relationship $AD$
over $D$ such that if $(D_j, D_i) \in AD$, then $D_j$ is an
ancestor of $D_i$, let $COST_R(D_i)$ be the access cost of
accessing unmaterialized $D_i$, and let $COST_A(D_i)$ be the
access cost of accessing materialized $D_i$. We have $COST_A(D_i)
< COST_R(D_i)$ since reconstruction of $D_i$ takes time. We use
$size(D_i)$ to denote the memory capacity required to store a
materialized XML element, $size(D_i) > 0$ and $size(D_i) <
size(D_j)$ for any $(D_j,D_i) \in AD$. Given a workload that is
characterized by $a_{i} (i=1,2,\ldots,n)$ representing the access
frequency of $D_i$. The \textit{XML reconstruction view selection
problem} is to select a set of elements $M$ from $D$ to be
materialized to minimize the total access cost

\[
\tau(D, M) = \sum_{i=1}^n a_{i} \times COST(D_i), \label{e1}
\]

under the disk capacity constraint

\[
\sum_{D_i \in M} size(D_i) \leq \delta,
\]

where $COST(D_i) = COST_A(D_i)$ if $D_i \in M$ or for some
ancestor $D_j$ of $D_i$, $D_j \in M$, otherwise $COST(D_i) =
COST_R(D_i)$. $\delta$ denotes the available memory capacity,
$\delta~\geq~0$.

Next, let $\bigtriangledown COST(D_i) = COST_R(D_i)- COST_A(D_i)$
means the cost saving by materialization, then one can show that
function $\tau$ is minimized if and only if the following function
$\lambda$ is maximized

\[
\lambda(D, M) = \sum_{D_i \in M^+} a_{i} \times \bigtriangledown
COST(D_i)
\]

under the disk capacity constraint

\[
\sum_{D_i \in M} size(D_i) \leq \delta,
\]

where $M^+$ represents all the materialized XML elements and their
descendant elements in $D$, it is defined as $M^+ = \{D_i \mid D_i
\in M\ or\ \exists D_j. (D_j, D_i) \in AD \wedge D_j~\in~M\}$.


\section{NP-Completeness}
\label{sec:complexity}

In this section, we prove that the XML reconstruction view
selection problem is NP-hard. First, the maximization problem is
changed into the equivalent decision problem.

\textbf{Equivalent decision problem.} Given $D$, $AD$,
$size(D_i)$, $a_i$, $\bigtriangledown COST(D_i)$ and $\delta$ as
defined in Section~\ref{sec:problem}, let $K$ denotes the cost
saving goal, $K~\geq~0$. Is there a subset $M \subseteq D$ such
that

\[
\sum_{D_i \in M^+} a_{i} \times \bigtriangledown COST(D_i) \geq K
\]

and

\[
\sum_{D_i \in M} size(D_i) \leq \delta
\]

$M^+$ represents all the materialized XML elements and their
descendant elements in $D$, it is defined as $M^+ = \{D_i \mid D_i
\in M\ or\ \exists D_j. (D_j, D_i) \in AD \wedge D_j~\in~M\}$.

In order to study this problem in a convenient model, we have the
following simplified version.

The input is a tree $T$, in which every node $i$ has a size $c_i$
and profit $p_i$, and the size limitation $C$. The target is to find
a subset of subtrees rooted at nodes $i_1,\cdots, i_k$ such that
$c_{i_1}+\cdots +c_{i_k}\le C$, and $p_{i_1}+\cdots +p_{i_k}$ is
maximal. Furthermore, there is no overlap between any two subtrees
selected in the solution.

We prove that the decision problem of the XML reconstruction view
selection is an NP-hard. A polynomial time reduction from
\textit{KNAPSACK}~\cite{book1979} to it is constructed.

\begin{theorem} \label{theoremNPcomplete}
The decision problem of the XML reconstruction view selection is
NP-complete.
\end{theorem}
\begin{proof}
It is straightforward to verify that the problem is in NP. Restrict
the problem to the well-known NP-complete problem
\textit{KNAPSACK}~\cite{book1979} by allowing only problem instances
in which:

Assume that a Knapsack problem has input $(p_1,c_1),\cdot,
(p_n,c_n)$, and parameters $K$ and $C$. We need to determine a
subset $S\subseteq \{1,\cdots, n\}$ such that $\sum_{i\in S} c_i\le
C$ and $\sum_{i\in S}p_i\ge K$.

Build a binary tree $T$ with exactly leaves. Let leaf $i$ have
profit $p_i$ and size $c_i$. Furthermore, each internal node, which
is not leaf, has size $\infty$ and profit $\infty$.

Clearly, any solution cannot contain any internal due to the size
limitation. We can only select a subset of leaves. This is
equivalent to the Knapsack problem.

\end{proof}

Finally, we state the NP-hardness of the XML reconstruction view
selection problem.

\begin{theorem}
The XML reconstruction view selection problem is NP-hard.
\end{theorem}
\begin{proof}
It follows from Theorem~\ref{theoremNPcomplete}, since the
equivalent decision problem is NP-complete.
\end{proof}

\section{Fully Polynomial-Time Approximation Scheme}
\label{sec:FPTAS}

We assume that each parameter is an integer. The input is $n$ XML
elements, $D = \{D_1, D_2, \cdots, D_n\}$ which will be
represented by an $AD$ tree $J$, where each edge in $J$ shows a
relationship between a pair of parent and child nodes.

We have a divide and conquer approach to develop a fully
approximation scheme. Given an $AD$ tree $J$ with root $r$, it has
subtrees $J_1,\cdots, J_k$ derived from the children $r_1,\cdots,
r_k$ of $r$. We find a set of approximate solutions among
$J_1,\cdots, J_{k/2}$ and another set of approximate solutions
among $J_{k/2+1},\cdots, J_{k}$.

We merge the two sets of approximate solutions to obtain the
solution for the union of subtrees $J_1,\cdots, J_k$. Add one more
solution that is derived by selecting the root $r$ of $J$. Group
those solutions into parts such that each part contains all
solutions $P$ with similar $\lambda(D,P)$. Prune those solution by
selecting the one from each part with the least size. This can
reduce the number of solution to be bounded by a polynomial.

We will use a list $P$ to represent the selection of elements from
$D$.

For a list of elements $P$, define $\lambda(P)=\sum_{D_i\in
P}a_i\times \nabla COST(D_i)$, and $\mu(P)=\sum_{D_i\in
P}size(D_i)$. Define $\chi(J)$ be the largest product of the node
degrees along a path from root to a leaf in the $AD$ tree $J$.

Assume that $\epsilon$ is a small constant with $1>\epsilon >0$. We
need an $(1+\epsilon)$ approximation. We maintain a list of
solutions $P_1, P_2,\cdots$, where $P_i$ is a list of elements in
$D$.

Let $f=(1+{\epsilon\over z})$ with $z=2\log \chi(J)$. Let
$w=\sum_{i=1}^n a_iCOST_R(D_i)$ and $s=\sum_{i=1}^n size(D_i)$.

Partition the interval $[0, w]$ into $I_1, I_2,\cdots, I_{t_1}$
such
 that $I_1=[0,1]$ and $I_k=(b_{k-1}, b_k]$ with $b_k=f\cdot b_{k-1}$ for $k<t$, and $I_{t_1}=(b_{t-1},w]$,
 where $b_{t-1}< w\le fb_{t-1}$.



Two lists $P_i$ and $P_j$,
 are in the same region if there exist $I_{k}$ such that both $\lambda(P_i)$ and $\lambda(P_j)$ are $I_k$.

For two lists of partial solutions $P_i=D_{i_1}\cdots D_{i_{m_1}}$
and $P_j=D_{j_1}\cdots D_{j_{m_2}}$, their link $P_i\circ
P_j=D_{i_1}\cdots D_{i_{m_1}}D_{j_1}\cdots D_{j_{m_2}}$.

\vskip 10pt

{\bf Prune} ($L$)

\qquad Input: $L$ is a list of partial solutions $P_1, P_2,
\cdots, P_m$;

\qquad Partition $L$ into parts $U_1, \cdots, U_v$ such that two
lists $P_i$ and $P_j$ are in the same part if $P_i$ and $P_j$ are
in the same region.

\qquad For each $U_i$, select $P_j$ such that $\mu(P_j)$ is the
least among all $P_j$ in $U_i$;

{\bf End of Prune}

\vskip 10pt

{\bf Merge} ($L_1, L_2$)

\qquad Input: $L_1$ and $L_2$ are two lists of solutions.

\qquad Let $L=\emptyset$;

\qquad For each $P_i\in L_1$ and each $P_j\in L_2$

\qquad\qquad append their link $P_i\circ P_j$ to $L$;

\qquad Return $L$;

{\bf End of Merge}

\vskip 10pt

{\bf Union} ($L_1, L_2,\cdots, L_k$)

\qquad Input: $L_1,\cdots, L_k$ are  lists of solutions.

\qquad If $k=1$ then return $L_1$;

\qquad Return Prune(Merge(Union$(L_1,\cdots, L_{k/2})$,

\qquad\qquad\qquad\qquad\qquad Union$(L_{k/2+1},\cdots, L_k)))$;

{\bf End of Union}

\vskip 10pt

{\bf Sketch} ($J$)

Input: $J$  is a set of $n$ elements according to their $AD$.

\qquad If $J$ only contains one element $D_i$, return the list
$L=D_i, \emptyset$ with two solutions.

\qquad Partition the list $J$ into subtrees $J_1,\cdots, J_k$
according to its $k$ children.

\qquad Let $L_0$ be the list that only contains solution $J$.

\qquad for $i=1$ to $k$ let $L_i$=Sketch($J_i$);

\qquad Return Union$(L_0, L_1,\cdots, L_k)$;

{\bf End of Sketch}

\vskip 10pt

For a list of elements $P$ and an $AD$ tree $J$, $P[J]$ is the list
of elements in both $P$ and $J$. If $J_1,\cdots, J_k$ are disjoint
subtrees of an $AD$ tree, $P[J_1,\cdots, J_k]$ is $P[J_1]\circ
\cdots \circ P[J_k]$.


In order to make it convenient, we make the tree $J$ normalized by
adding some useless nodes $D_i$ with $COST_R(D)=COST_A(D)=0$. The
size of tree is at most doubled after normalization. In the rest
of the section, we always assume $J$ is normalized.

\begin{lemma}\label{lemma-union} Assume that $L_i$  is a list of solutions for
the problem with $AD$ tree $J_i$ for $i=1,\cdots, k$. Let $P_i\in
L_i$ for $i=1,\cdots, k$. Then there exists $P\in
L=$Union($L_1,\cdots, L_k$) such that $\lambda(P)\le f^{\log k}\cdot
\lambda(P_1\circ \cdots \circ P_k)$ and $\mu(P^*)\le \mu(P_1\circ
\cdots \circ P_k))$.
\end{lemma}

\begin{proof}
We prove by induction. It is trivial when $k=1$. Assume that the
lemma is true for cases less than $k$.

Let $M_1$ $=$ Union$(L_1,\cdots, L_{k/2})$ and $M_2$ $=$
Union$(L_{k/2+1},\cdots, L_{k})$.

Assume that $M_1$ contains $Q_1$ such that $\lambda(Q_1)\le f^{\log
(k/2)}\lambda(P_1\circ \cdots\circ P_{k/2})$ and $\mu(Q_1)\le
\mu(P_1\circ\cdots\circ P_{k/2})$.

Assume that $M_2$ contains $Q_2$ such that $\lambda(Q_2)\le f^{\log
(k/2)}\lambda(P_{k/2+1}\circ\cdots\circ P_{k}])$ and $\mu(Q_2)\le
\mu(P_{k/2+1}\circ\cdots\circ P_{k})$.

Let $Q=Q_1\circ Q_2$. Let $Q^*$ be the solution in the same region
with $Q$ and has the least $\mu(Q^*)$.  Therefore,

\begin{eqnarray*}
&&\lambda(Q^*)\\
&\le& f\lambda(Q)\\
&\le& f^{\log
(k/2)}f\lambda(P_1\circ\cdots\circ P_{k})\\
& \le& f^{\log k}\lambda(P_1\circ\cdots\circ P_{k}).
\end{eqnarray*}

Since $\mu(Q_1)\le \mu(P_1\circ\cdots\circ P_{k/2})$ and
$\mu(Q_2)\le \mu(P_{k/2+1}\circ\cdots\circ P_{k})$, we also have

\begin{eqnarray*}
\mu(Q^*)&\le& \mu(Q)\\
&\le& \mu(Q_1)+\mu(Q_2)\\
&\le& \mu(P_1\circ \cdots\circ
P_{k/2})+\mu(P_{k/2+1}\circ\cdots\circ P_{k})\\
&=&\mu(P_1\circ \cdots\circ
P_{k})\\
&=&\mu(P).
\end{eqnarray*}
\end{proof}

\begin{lemma}\label{lemma-1} Assume that $P$  is an arbitrary solution for
the problem with $AD$ tree $J$. For $L$$=$Sketch($J$), there exists
a solution $P'$ in the list $L$ such that $\lambda(P')\le f^{\log
\chi(J)}\cdot \lambda(P)$ and $\mu(P')\le \mu(P)$.
\end{lemma}

\begin{proof} We prove by induction.
The basis at $|J|\le 1$ is trivial. We assume that the claim is
true for all  $|J|< m$. Now assume that $|J|=m$ and $J$ has $k$
children which induce $k$ subtrees $J_1,\cdots, J_k$.

Let $L_i=P[J_i]$ for $i=1,\cdots, k$. By our hypothesis, for each
$i$ with $1\le i\le k$, there exists $Q_i\in L_i$ such that
$\lambda(Q_i)\le f^{\log \chi(J_i)}\cdot \lambda(P[J_i])$ and
$\mu(Q_i)\le \mu(P[J_i]))$.

Let $M$$=$Union$(L_1,\cdots, L_{k})$. By Lemma~\ref{lemma-union},
there exists $P'\in M$ such that

\begin{eqnarray*}
\lambda(P')&\le& f^{\log (k/2)} \lambda(Q_1\circ \cdots\circ Q_k)\\
&\le& f^{\max\{\log \chi(J_1),\cdots, \log
\chi(J_{k/2})\}}\\
&&f^{\log (k/2)}\lambda(P[J_1,\cdots, J_{k}])\\
 & \le& f^{\log
\chi(J)}\lambda(P[J_1,\cdots, J_{k}])\\
&=& f^{\log \chi(J)}\lambda(P).
\end{eqnarray*}

 and
\begin{eqnarray*}
 \mu(P')&\le&\mu(Q_1\circ\cdots\circ Q_k)\\
 &\le& \mu(P[J_1,\cdots, J_{k}])\\
 &=& \mu(P).
 \end{eqnarray*}
\end{proof}

\begin{lemma}\label{time-lemma}
Assume that $\mu(D, J)\le a(n)$. Then the computational time for
Sketch($J$) is $O(|J|({(\log \chi(J))(\log a(n))\over
\epsilon})^{2})$, where $|J|$ is the number of nodes in $J$.
\end{lemma}

\begin{proof} The number of intervals is $O({(\log \chi(J))(\log a(n))\over
\epsilon})$. Therefore the list of each $L_i$$=$Prune($J_i$) is of
length $O({(\log \chi(J))(\log a(n))\over \epsilon})$.

Let $F(k)$ be the time for Union$(L_1,\cdots, L_k)$. It satisfies
the recursion $F(k)=2F(k/2)+O(({(\log \chi(J))(\log a(n))\over
\epsilon})^2)$. This brings solution $F(k)=O(k({(\log
\chi(J))(\log a(n))\over \epsilon})^2)$.

 Let $T(J)$ be
the computational time for Prune($J$).

Denote $E(J)$ to be the number of edges in $J$.

 We prove by induction that
$T(J)\le c E(J)({(\log \chi(J))(\log a(n))\over \epsilon})^2$ for
some constant $c>0$. We select constant $c$ enough so that merging
two lists takes $c(n\log a(n))^2)$ steps.

 We have that
\begin{eqnarray*}
T(J)&\le& T(J_1)+\cdots +T(J_k)+F(k)\\
&\le& cE(J_1)({(\log \chi(J))(\log a(n))\over \epsilon})^2+\\
&&\cdots+cE(J_k)({(\log \chi(J))(\log a(n))\over \epsilon})^2\\
&&+ck({(\log \chi(J))(\log a(n))\over \epsilon})^2\\
&\le& cE(J)({(\log \chi(J))(\log a(n))\over \epsilon})^2\\
&\le& c|J|({(\log \chi(J))(\log a(n))\over \epsilon})^2.
\end{eqnarray*}
\end{proof}

We now complete the common procedure. Before we give the FPTAS for
our problems, we first give the following lemma which facilitates
our proof for FPTAS. One can refer the algorithm book
\cite{CormenLeisersonRivestStein01} for its proof.

\begin{lemma}\label{basic-lemma}
(1)For $0\le x\le 1$, $e^x\le 1+x+x^2$.

(2)For real $y\ge 1$, $(1+{1\over y})^y\le e$.
\end{lemma}

\vskip 10pt

{\bf  Algorithm}

Approximate($J$, $\epsilon$)

Input: $J$ is an $AD$ tree with elements $D_1,\cdots, D_n$ and
$\epsilon$ is a small constant with $1>\epsilon>0$;

\qquad Let $L=$Sketch$(J)$;

\qquad Select $P_i$ from the list $L$ that $P_i$ has the optimal
cost;

{\bf End of the Algorithm}

\vskip 10pt
\begin{theorem}\label{main-theorem}
For any instance of $J$ of an $AD$ tree with $n$ elements, there
exists an $O(n({(\log \chi(J))(\log a(n))\over \epsilon})^{2})$
time approximation scheme, where $\sum_{i=1}^n a_iCOST_R(D_i)\le
a(n)$.
\end{theorem}

\begin{proof}
Assume that $P$ is the optimal solution for input $J$. Let
$L$$=$Prune($J$). By Lemma~\ref{lemma-1}, we have $P^*\in L$ that
satisfies the condition of Lemma~\ref{lemma-1}.
\begin{eqnarray*}
\lambda(P^*)&\le& f^{\log \chi(J)}\lambda(P) \\
&=& (1+{\epsilon\over 2 z })^{\log \chi(J)} \lambda(P) \\
&\le& e^{\epsilon\over 2}\cdot \lambda(P) \\
&=& (1+{\epsilon\over 2}+({\epsilon\over 2})^2)\cdot \lambda(P) \\
&<& (1+\epsilon)\cdot \lambda(P) .
\end{eqnarray*}

Furthermore, $\mu(P^*)\le \mu(P)$.

 The computational time follows
from Lemma~\ref{time-lemma}.
\end{proof}

It is easy to see that $\chi(J)\le 2^{|J|}$. We have the following
corollary.
\begin{corollary}
For any instance of $J$ of an $AD$ tree with $n$ elements, there
exists an $O(n^3({\log a(n)\over \epsilon})^{2})$ time
approximation scheme, where $\sum_{i=1}^n a_iCOST_R(D_i)\le a(n)$.
\end{corollary}

\section{Extension to Multiple Trees}
In this section, we show an approximation scheme for the problem
with an input of multiple trees. The input is a series of trees
$J_1,\cdots, J_k$.

\vskip 10pt
\begin{theorem}
For any instance of $J$ of an $AD$ tree with $n$ elements, there
exists an $O(n({(\log \chi(J_0))(\log a_0(n))\over \epsilon})^{2})$
time approximation scheme, where $\sum_{j}(p_j+c_j)\le a_0(n)$ and
$J_0$ is a tree via connecting all $J_1,\cdots, J_k$ into a single
tree under a common root $r_0$.
\end{theorem}

\begin{proof}
Build a new tree with a new node $r_0$ such that $J_1,\cdots, J_k$
are the subtrees under $r_0$.

Apply the algorithm in in Theorem~\ref{main-theorem}.
\end{proof}

\section{Conclusion and Future Work}
\label{sec:conclude}

In this work, we studied the problem of XML reconstruction view
selection that promises to improve query evaluation in XML
databases. We were first to formally define this problem: given a
set of XML elements $D$ from an XML database, their access
frequencies $a_i$ (aka workload), a set of ancestor-descendant
relationships $AD$ among these elements, and a storage capacity
$\delta$, find a set of elements $M$ from $D$, whose XML subtrees
should be materialized as reconstruction views, such that their
combined size is no larger than $\delta$. Next, we showed that the
XML reconstruction view selection problem is NP-hard. Finally, we
proposed a fully polynomial-time approximation scheme (FPTAS) that
can be used to solve the problem in practice.

Future work for our research includes two main directions: (1)~an
extension of the proposed solution to support multiple XML trees
and (2)~an implementation and performance study of our framework
in an existing XML database.



\begin{thebibliography}{10}

\bibitem{BunemanSA99}
S.~Abiteboul, P.~Buneman, and D.~Suciu.
\newblock {\em Data on the Web: From Relations to Semistructured Data and XML}.
\newblock Morgan Kaufmann, 1999.

\bibitem{AtayCLLF07}
M.~Atay, A.~Chebotko, D.~Liu, S.~Lu, and F.~Fotouhi.
\newblock Efficient schema-based {XML}-to-relational data mapping.
\newblock {\em Inf. Syst.}, 32(3):458--476, 2007.

\bibitem{BalminOBCP04}
A.~Balmin, F.~{\"O}zcan, K.~S. Beyer, R.~Cochrane, and H.~Pirahesh.
\newblock A framework for using materialized {XPath} views in {XML} query
  processing.
\newblock In {\em VLDB}, pages 60--71, 2004.

\bibitem{BalminP05}
A.~Balmin and Y.~Papakonstantinou.
\newblock Storing and querying {XML} data using denormalized relational
  databases.
\newblock {\em VLDB J.}, 14(1):30--49, 2005.

\bibitem{BaralisPT97}
E.~Baralis, S.~Paraboschi, and E.~Teniente.
\newblock Materialized views selection in a multidimensional database.
\newblock In {\em VLDB}, pages 156--165, 1997.

\bibitem{BohannonFRS02}
P.~Bohannon, J.~Freire, P.~Roy, and J.~Sim{\'e}on.
\newblock From {XML} schema to relations: A cost-based approach to {XML}
  storage.
\newblock In {\em ICDE}, pages 64--75, 2002.

\bibitem{BonczGKMRT06}
P.~A. Boncz, T.~Grust, M.~van Keulen, S.~Manegold, J.~Rittinger, and
  J.~Teubner.
\newblock {MonetDB/XQuery}: a fast {XQuery} processor powered by a relational
  engine.
\newblock In {\em SIGMOD Conference}, pages 479--490, 2006.

\bibitem{ChaudhuriCSW05}
S.~Chaudhuri, Z.~Chen, K.~Shim, and Y.~Wu.
\newblock Storing {XML} (with {XSD}) in {SQL} databases: Interplay of logical
  and physical designs.
\newblock {\em IEEE Trans. Knowl. Data Eng.}, 17(12):1595--1609, 2005.

\bibitem{ChebotkoALF07}
A.~Chebotko, M.~Atay, S.~Lu, and F.~Fotouhi.
\newblock {XML} subtree reconstruction from relational storage of {XML}
  documents.
\newblock {\em Data Knowl. Eng.}, 62(2):199--218, 2007.

\bibitem{ChebotkoLALF05}
A.~Chebotko, D.~Liu, M.~Atay, S.~Lu, and F.~Fotouhi.
\newblock Reconstructing {XML} subtrees from relational storage of {XML}
  documents.
\newblock In {\em ICDE Workshops}, page 1282, 2005.

\bibitem{ChirkovaHS02}
R.~Chirkova, A.~Y. Halevy, and D.~Suciu.
\newblock A formal perspective on the view selection problem.
\newblock {\em VLDB J.}, 11(3):216--237, 2002.

\bibitem{CormenLeisersonRivestStein01}
T.~H. Cormen, C.~E. Leiserson, R.~L. Rivest, and C.~Stein.
\newblock {\em Introduction to Algorithms, Second Edition}.
\newblock The MIT Press, 2001.

\bibitem{FanGJK07}
W.~Fan, F.~Geerts, X.~Jia, and A.~Kementsietsidis.
\newblock Rewriting regular {XPath} queries on {XML} views.
\newblock In {\em ICDE}, pages 666--675, 2007.

\bibitem{FlorescuK99}
D.~Florescu and D.~Kossmann.
\newblock Storing and querying {XML} data using an {RDMBS}.
\newblock {\em IEEE Data Eng. Bull.}, 22(3):27--34, 1999.

\bibitem{book1979}
M.~R. Garey and D.~S. Johnson.
\newblock {\em Computer and Intractability: A Guide to the Theory of
  NP-Completeness.}
\newblock W. H. Freeman, 1979.

\bibitem{Gupta97}
H.~Gupta.
\newblock Selection of views to materialize in a data warehouse.
\newblock In {\em ICDT}, pages 98--112, 1997.

\bibitem{GuptaM99}
H.~Gupta and I.~S. Mumick.
\newblock Selection of views to materialize under a maintenance cost
  constraint.
\newblock In {\em ICDT}, pages 453--470, 1999.

\bibitem{HundlingSW05}
J.~H{\"u}ndling, J.~Sievers, and M.~Weske.
\newblock {NaXDB} - realizing pipelined {XQuery} processing in a native {XML}
  database system.
\newblock In {\em XIME-P}, 2005.

\bibitem{JagadishACLNPPSWWY02}
H.~V. Jagadish, S.~Al-Khalifa, A.~Chapman, L.~V.~S. Lakshmanan,
A.~Nierman,
  S.~Paparizos, J.~M. Patel, D.~Srivastava, N.~Wiwatwattana, Y.~Wu, and C.~Yu.
\newblock {TIMBER}: A native {XML} database.
\newblock {\em VLDB J.}, 11(4):274--291, 2002.

\bibitem{KarloffM99}
H.~J. Karloff and M.~Mihail.
\newblock On the complexity of the view-selection problem.
\newblock In {\em PODS}, pages 167--173, 1999.

\bibitem{Kossmann00}
D.~Kossmann.
\newblock The state of the art in distributed query processing.
\newblock {\em ACM Comput. Surv.}, 32(4):422--469, 2000.

\bibitem{KrishnamurthyCKN04}
R.~Krishnamurthy, V.~T. Chakaravarthy, R.~Kaushik, and J.~F.
Naughton.
\newblock Recursive {XML} schemas, recursive {XML} queries, and relational
  storage: {XML-to-SQL} query translation.
\newblock In {\em ICDE}, pages 42--53, 2004.

\bibitem{LakshmananWZ06}
L.~V.~S. Lakshmanan, H.~Wang, and Z.~J. Zhao.
\newblock Answering tree pattern queries using views.
\newblock In {\em VLDB}, pages 571--582, 2006.

\bibitem{MandhaniS05}
B.~Mandhani and D.~Suciu.
\newblock Query caching and view selection for {XML} databases.
\newblock In {\em VLDB}, pages 469--480, 2005.

\bibitem{SawiresTPAAC06}
A.~Sawires, J.~Tatemura, O.~Po, D.~Agrawal, A.~E. Abbadi, and K.~S.
Candan.
\newblock Maintaining {XPath} views in loosely coupled systems.
\newblock In {\em VLDB}, pages 583--594, 2006.

\bibitem{TatemuraSPAC05}
A.~Sawires, J.~Tatemura, O.~Po, D.~Agrawal, and K.~S. Candan.
\newblock Incremental maintenance of path expression views.
\newblock In {\em SIGMOD Conference}, pages 443--454, 2005.

\bibitem{Schoning01}
H.~Sch{\"o}ning.
\newblock Tamino - a {DBMS} designed for {XML}.
\newblock In {\em ICDE}, pages 149--154, 2001.

\bibitem{ShanmugasundaramGTZDN99}
J.~Shanmugasundaram, K.~Tufte, C.~Zhang, G.~He, D.~J. DeWitt, and
J.~F.
  Naughton.
\newblock Relational databases for querying {XML} documents: Limitations and
  opportunities.
\newblock In {\em VLDB}, pages 302--314, 1999.

\bibitem{TangYOCW08}
N.~Tang, J.~X. Yu, M.~T. {\"O}zsu, B.~Choi, and K.-F. Wong.
\newblock Multiple materialized view selection for {XPath} query rewriting.
\newblock In {\em ICDE}, pages 873--882, 2008.

\bibitem{TangYTOB09}
N.~Tang, J.~X. Yu, H.~Tang, M.~T. {\"O}zsu, and P.~A. Boncz.
\newblock Materialized view selection in {XML} databases.
\newblock In {\em DASFAA}, pages 616--630, 2009.

\bibitem{XuO05}
W.~Xu and Z.~M. {\"O}zsoyoglu.
\newblock Rewriting {XPath} queries using materialized views.
\newblock In {\em VLDB}, pages 121--132, 2005.

\bibitem{YangKL97}
J.~Yang, K.~Karlapalem, and Q.~Li.
\newblock Algorithms for materialized view design in data warehousing
  environment.
\newblock In {\em VLDB}, pages 136--145, 1997.

\end{thebibliography}

\end{document}